\documentclass[aps,prl,twocolumn,superscriptaddress]{revtex4-2}
\usepackage[utf8]{inputenc}
\usepackage[T1]{fontenc}
\usepackage{lmodern}
\usepackage{amsmath,amsthm,amssymb,amsfonts}
\usepackage{graphicx,bm}
\usepackage{hyperref}
\usepackage{physics}

\hypersetup{
    colorlinks=true,
    linkcolor=magenta,
    citecolor=blue,
    urlcolor=blue}

\newtheorem{theorem}{Theorem}

\begin{document}

\author{Chakir Fikri}
\affiliation{Independent Researcher}
\email{chakir.fikri.research@proton.me}

\date{\today}

\title{Structural constraint on delayed-choice quantum eraser architectures}

\begin{abstract}
Delayed-choice quantum eraser (DCQE) experiments are often presented as challenging classical causal intuitions by correlating detection events with choices implemented at later times. While it is well understood that post-selection plays a crucial role in producing the observed interference patterns, the structural features underlying such correlations are rarely analyzed within a unified framework. In this work, we introduce a simple probabilistic constraint applicable to idealized DCQE architectures. We show that four intuitive properties—statistical independence of the choice, absence of losses, deterministic routing conditioned on the choice, and distinct conditional detection distributions—cannot be simultaneously satisfied. This incompatibility provides a transparent classification of DCQE schemes and clarifies how conditional interference patterns arise without invoking exotic mechanisms.
\end{abstract}

\maketitle

\section{Introduction}

Quantum entanglement enables correlations between distant or temporally separated events that challenge classical causal intuitions. In spatially separated systems, such correlations are quantitatively characterized by the violation of Bell inequalities \cite{Bell1964, Aspect1982, Gisin1998, Hensen2015, bigbell2018}. In temporal scenarios, delayed-choice experiments explore how measurement outcomes appear correlated with experimental choices implemented at later times \cite{Jacques2007}. 

Among these, delayed-choice quantum eraser (DCQE) experiments occupy a prominent place. In a typical DCQE setup, two entangled particles are generated simultaneously. One particle is detected at a spatially resolved detector, while the other undergoes a measurement whose configuration—often interpreted as erasing or preserving which-path information—is chosen after the first detection event \cite{Scully1982, Kim2000, Ma2016}. When coincidence counts are analyzed, interference patterns are observed for certain subsets of outcomes and disappear for others.

Although it is well established that such patterns arise from conditional statistics and post-selection rather than any form of retrocausal influence \cite{Ma2016, SEP_Retrocausality}, DCQE experiments are still frequently discussed in language suggesting a tension with standard causal explanations. This motivates a closer examination of the structural assumptions implicitly attributed to idealized DCQE scenarios, independently of any specific physical implementation.

More precisely, we argue that the appearance of paradox in DCQE scenarios results from implicitly combining several intuitive structural assumptions that cannot, in fact, be simultaneously satisfied. From this perspective, the apparent tension associated with DCQE experiments can be understood as a structural illusion arising from mutually incompatible expectations rather than as evidence for any unconventional causal mechanism.

Several authors have argued that the apparent paradoxical features of DCQE experiments disappear once the physical and causal structure of the experiment is analyzed explicitly. 
Kastner \cite{Kastner2019} emphasizes that DCQE experiments are structurally equivalent to standard EPR-type correlation experiments and therefore involve neither a genuine erasure of previously available which-path information nor any delayed physical influence. 
Fankhauser \cite{Fankhauser2020} shows that the correlations observed in DCQE experiments can be consistently described within both collapse interpretations of quantum mechanics and de Broglie–Bohm theory without invoking any backward-in-time causal influence. 
Waaijer and van Neerven \cite{Waaijer2024} provides a step-by-step forward-time analysis of delayed-choice experiments showing that delayed and non-delayed configurations lead to the same final quantum state within standard quantum mechanics, thereby removing the appearance of retrocausal behavior at the operational level without requiring interpretational assumptions about the ontic or epistemic status of the quantum state. 
Qureshi \cite{Qureshi2020, Qureshi2021, Qureshi2025} has argued that the apparent availability of which-path information in delayed configurations may already be constrained by the structure of the joint entangled state prior to the final detection stage, so that the delayed-choice configuration does not correspond to a genuine late-time operational selection between complementary measurement contexts.

The present work addresses a complementary question. Rather than focusing on the temporal interpretation of delayed-choice operations or on the status of which-path information prior to coincidence post-selection, we investigate minimal probabilistic compatibility relations between structural properties commonly attributed to idealized DCQE scenarios expressed in terms of observable random variables describing the experimental architecture. The resulting incompatibility applies independently of how the delayed-choice configuration is interpreted at the level of conditional quantum states and does not rely on the dynamical structure of quantum mechanics itself. Instead, it follows from elementary probabilistic constraints relating experimentally accessible variables.

Specifically, we identify four conditions that formalize common expectations: statistical independence between the experimental choice and the earlier detection event, absence of losses, deterministic association between choices and detection outcomes, and the existence of distinct conditional detection distributions corresponding to interference and non-interference patterns. We show that these conditions are mutually incompatible.

This simple structural constraint does not rely on quantum dynamics or specific experimental details. Instead, it provides a unifying perspective that clarifies why realistic DCQE implementations necessarily violate at least one of these intuitive properties. We then apply this framework to several representative DCQE architectures, illustrating how losses, coarse-graining, or nondeterministic routing account for the observed correlations.

Our analysis is not intended to resolve interpretational debates, but rather to provide a transparent classification of DCQE schemes and to identify the minimal structural ingredients responsible for their behavior. In this way, it clarifies why the apparent paradoxical features of DCQE experiments arise naturally in idealized descriptions but disappear once the underlying compatibility constraints are made explicit.

\section{Formal framework and definitions}

\subsection{Probabilistic framework}

We consider a generic DCQE scenario involving two correlated systems, denoted $S$ and $I$. System $S$ is detected at a spatially resolved detector (or an equivalent scanning detection device), producing a random variable $X$ that labels the detection position. System $I$ is subjected to a measurement configuration chosen at a later time and subsequently detected.

To describe such experiments at a minimal probabilistic level, we introduce three random variables defined on a common probability space:

\begin{itemize}
    \item $X$, representing the detection outcome of system $S$ at the spatially resolved detector;
    \item $C$, representing the experimental choice implemented on system $I$. For simplicity, we take $C \in \{\mathrm{erase},\mathrm{preserve}\}$;
    \item $D$, representing the detection outcome associated with system $I$.
\end{itemize}

In realistic implementations, multiple detectors may be associated with each experimental configuration. It is therefore useful to distinguish between fine-grained detection outcomes, corresponding to individual physical detectors, and coarse-grained outcomes that group detectors according to the experimental choice they are intended to implement. Throughout this work, $D$ may denote either level of description, with the relevant coarse-graining explicitly stated when necessary.

We now formalize four intuitive properties that are often implicitly attributed to idealized DCQE scenarios.
\begin{itemize}
\item \textit{Independence of the experimental choice:} 
$X \perp C$. This condition expresses the assumption that the choice of measurement configuration applied to system $I$ is statistically independent of the earlier detection outcome of system $S$. Operationally, it corresponds to the usual notion of a freely chosen experimental setting.

\item \textit{Absence of losses:}
All realizations of $S$ are assumed to be paired with a detection event of $I$. In probabilistic terms, $X$ and $D$ are jointly defined for every trial. No post-selection based on missing detection events is required.

\item \textit{Deterministic routing conditioned on the choice:}
We assume that the detection outcome of system $I$ is uniquely determined by the experimental choice: $D = f(C)$, for some deterministic function $f$. At the coarse-grained level, this corresponds to the idealized picture in which each choice $\mathrm{erase}$ or $\mathrm{preserve}$ leads to a distinct and exclusive detection channel.

\item \textit{Distinct conditional detection distributions:}
Finally, we assume that the detection statistics of $S$ conditioned on the outcome of $I$ differ for at least two detection outcomes: $P(X \in A \mid D = d) \neq P(X \in A \mid D = d')$, for some measurable set $A$ and some pair of outcomes $d \neq d'$. In DCQE experiments, this typically corresponds to the appearance of interference fringes for certain coincidence counts and their absence for others.
\end{itemize}

For convenience, we refer to a DCQE architecture satisfying all four properties above as an \emph{idealized DCQE}. These conditions formalize a simplified intuitive picture in which a freely chosen experimental configuration cleanly separates interference and non-interference outcomes without losses.
As we show in the following section, such an idealized scenario is structurally impossible: the four properties cannot be simultaneously satisfied.

\subsection{Relation to the Hilbert-space description of DCQE implementations}

Although the structural constraint derived in this work is formulated at a purely probabilistic level, the four properties introduced above admit natural counterparts in the Hilbert-space description of DCQE experiments.

In standard DCQE realizations, the joint system $(S,I)$ is described by a density operator $\rho$ evolving through unitary transformations followed by projective measurements associated with different detection configurations. In this framework, deterministic routing conditioned on the choice corresponds to the idealized situation in which the measurement operators associated with distinct configurations project onto orthogonal subspaces of the detector Hilbert space. Operationally, this means that detection events corresponding to different choices are described by mutually orthogonal projectors $\{\Pi_d\}$ satisfying $\Pi_d \Pi_{d'} = 0$ for $d \neq d'$.

Similarly, loss-induced post-selection arises when the erasure configuration involves projection onto a proper subspace of the available detection space. In this case, only a subset of incoming events contributes to coincidence statistics. At the level of the density operator, this corresponds to conditioning on the action of a non-unitary measurement operator $M$ with $M^\dagger M < \mathbb{I}$, so that the post-selected state takes the form
\[
\rho' = \frac{M \rho M^\dagger}{\mathrm{Tr}(M \rho M^\dagger)}.
\]

The existence of distinct conditional detection distributions is closely related to the presence of off-diagonal coherence terms in the joint density matrix prior to conditioning. 
More explicitly, consider a joint $(S,I)$ state of the form
\[
\rho = |\Psi\rangle\langle\Psi|, \qquad
|\Psi\rangle = \frac{1}{\sqrt{2}}\bigl(|s_1\rangle|i_1\rangle + |s_2\rangle|i_2\rangle\bigr).
\]
If the path states of system $I$ are orthogonal, $\langle i_1 | i_2 \rangle = 0$, tracing over system $I$ yields a reduced density operator for system $S$,
\[
\rho_s = \mathrm{Tr}_i(\rho)
= \frac{1}{2}\bigl(|s_1\rangle\langle s_1| + |s_2\rangle\langle s_2|\bigr),
\]
in which the off-diagonal coherence terms vanish. This accounts for the absence of interference in the marginal distribution of system $S$ independently of any coincidence post-selection.
In particular, interference patterns observed in coincidence counts arise from non-vanishing matrix elements between path subspaces that are recombined within the post-selected detection channel. If these off-diagonal terms vanish due to decoherence or path distinguishability, the conditional asymmetries analyzed in the present work disappear. Interference patterns therefore require that coherence between the relevant path subspaces be preserved prior to conditioning.

From this perspective, the probabilistic incompatibility identified here should be understood as a structural constraint applying to measurement architectures operating within a coherent Hilbert-space framework. When coherence between relevant path subspaces is preserved prior to conditioning, distinct conditional detection distributions become possible, but only at the price of relaxing at least one of the other structural assumptions identified above.

\section{Structural constraint}

We now establish a simple probabilistic incompatibility between the properties introduced in the previous section.

\begin{theorem}[No-go result]
Let $X$, $C$, and $D$ be random variables defined on a common probability space. The following three statements cannot be simultaneously satisfied:

(i) $X$ and $C$ are statistically independent: $X \perp C$;

(ii) $D$ is a deterministic function of $C$: $D = f(C)$;

(iii) There exist a measurable set $A$ and two distinct outcomes $d \neq d'$ such that $P(X \in A \mid D = d) \neq P(X \in A \mid D = d')$.
\end{theorem}

\begin{proof} If $X$ and $C$ are independent, then $X$ is independent of any deterministic function of $C$. In particular, $X \perp f(C)$. Using $D = f(C)$, this implies $X \perp D$. Consequently, for any measurable set $A$ and any outcomes $d,d'$, one has $P(X \in A \mid D = d) = P(X \in A \mid D = d')$, which contradicts condition (iii).
\end{proof}

The theorem immediately implies that the four properties defining an idealized DCQE scenario cannot all hold simultaneously. Indeed, independence of the choice corresponds to condition (i), deterministic routing to condition (ii), and distinct conditional detection distributions to condition (iii). The absence of losses ensures that all random variables are jointly defined on the same probability space without conditioning on detection events.

We therefore conclude that any physically realizable DCQE architecture must necessarily violate at least one of the following intuitive assumptions:

\begin{itemize}
    \item statistical independence between $X$ and $C$,
    \item absence of losses,
    \item deterministic association between the choice and detection outcome,
    \item existence of distinct conditional detection distributions.
\end{itemize}

This structural constraint is purely probabilistic and does not rely on quantum dynamics or specific experimental details. It provides a minimal framework for analyzing how different DCQE implementations produce their characteristic correlations.

\section{Application to representative DCQE architectures}

We now illustrate how the structural constraint manifests itself in several representative DCQE implementations. Each architecture realizes three of the four intuitive properties introduced in the second section, while necessarily violating the remaining one.

\subsection{Architecture I: Coarse-graining and interference cancellation}

\begin{figure}[h]
    \centering
    \includegraphics[width=0.85\columnwidth]{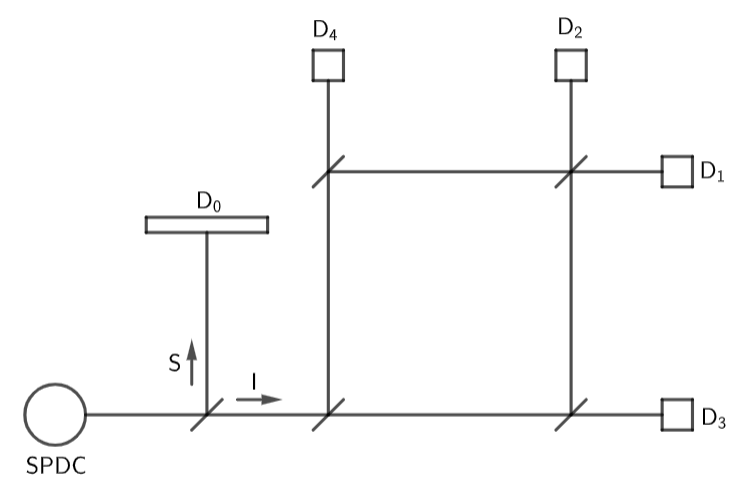}
    \caption{Example of architecture I.}
    \label{fig:ArchitectureI}
\end{figure}

A paradigmatic example is provided by the experiment of Kim \textit{et al.} \cite{Kim2000}. In this setup [\ref{fig:ArchitectureI}], system $I$ is detected by four detectors, denoted $D_1, D_2, D_3,$ and $D_4$. Two of these detectors are associated with the erasure configuration, while the other two correspond to the preservation of which-path information. At the coarse-grained level, one defines two outcomes:
\[
D_{\mathrm{erase}} = \{D_1,D_2\}, \quad D_{\mathrm{preserve}} = \{D_3,D_4\}.
\]
In principle, the evolution is unitary and no losses occur, satisfying the absence-of-losses condition. Moreover, at the coarse-grained level, the detection outcome is deterministically associated with the experimental choice, thus satisfying the deterministic routing property. However, when conditional detection statistics of $S$ are computed with respect to the coarse-grained outcomes, the interference patterns corresponding to $D_1$ and $D_2$ cancel due to their relative phase shift. As a result, one finds
\[
P(X \mid D_{\mathrm{erase}}) = P(X \mid D_{\mathrm{preserve}}),
\]
so that the distinct conditional distributions property is violated at this level of description.

If one instead conditions on the fine-grained outcomes $D_1$ and $D_2$ separately, interference patterns reappear, but the deterministic routing condition is then violated, since both outcomes correspond to the same experimental choice.

\subsection{Architecture II: Nondeterministic routing}

\begin{figure}[h]
    \centering
    \includegraphics[width=0.95\columnwidth]{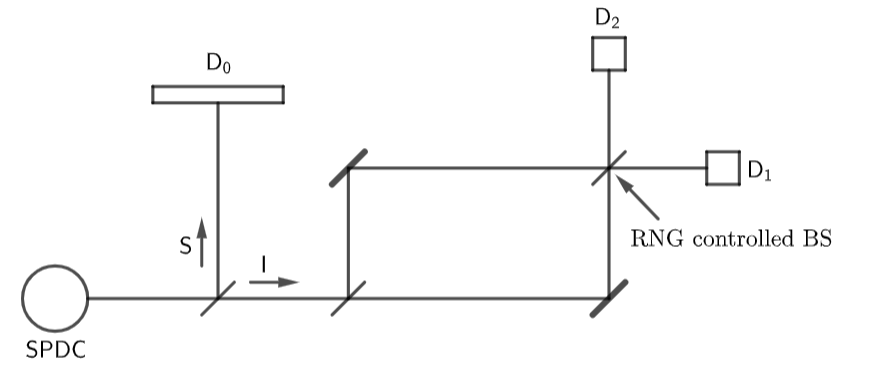}
    \caption{Example of architecture II.}
    \label{fig:ArchitectureII}
\end{figure}

A second class of implementations is exemplified by delayed-choice interferometric experiments in which system $I$ traverses a Mach-Zehnder interferometer whose final beam splitter is inserted or removed according to the experimental choice \cite{Jacques2007}.

In such setups [\ref{fig:ArchitectureII}], the evolution remains unitary and there are no losses in principle. The detection of $I$ occurs at two output detectors, denoted $D_1$ and $D_2$. When coincidence counts with system $S$ are analyzed, the conditional distributions $P(X \mid D_1)$ and $P(X \mid D_2)$ exhibit interference patterns (with reduced visibility) with opposite phases, satisfying the distinct conditional distributions property. However, each experimental choice—corresponding to inserting or removing the final beam splitter—leads to both detection outcomes with nonzero probability. The association between the choice and detection outcome is therefore nondeterministic, violating the deterministic routing condition.
It is important to note that the appearance of complementary interference patterns in the conditional distributions $P(X \mid D_1)$ and $P(X \mid D_2)$ relies on the preservation of coherence between the interferometric paths prior to the final recombination stage. If path distinguishability becomes complete before recombination, the corresponding off-diagonal terms vanish and the conditional interference structure disappears. The structural constraint discussed here therefore applies specifically to configurations in which coherent recombination is maintained up to detection.

\subsection{Architecture III: Loss-induced post-selection}

\begin{figure}[h]
    \centering
    \includegraphics[width=0.95\columnwidth]{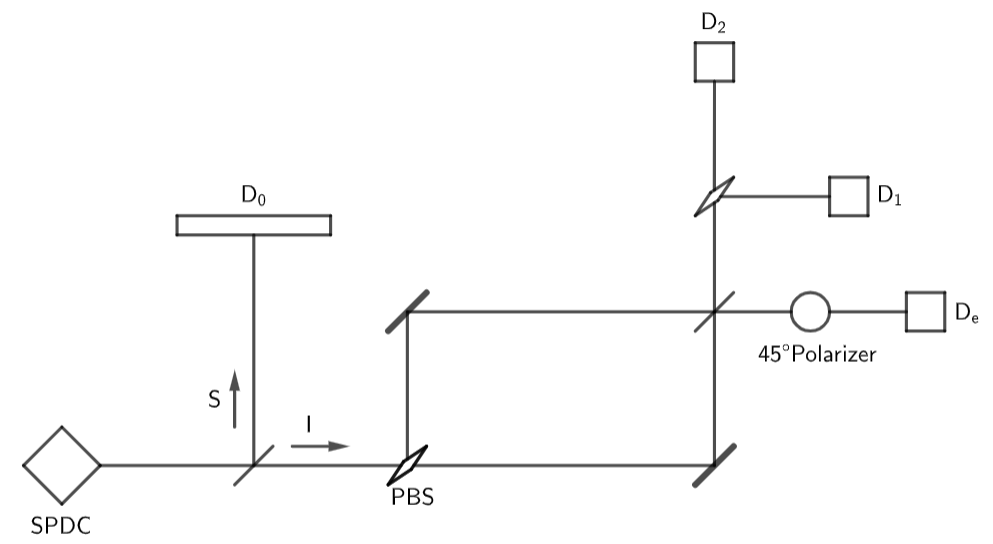}
    \caption{Example of architecture III.}
    \label{fig:ArchitectureIII}
\end{figure}

A third representative architecture employs entanglement in internal degrees of freedom such as polarization to mark and erase which-path information \cite{Ma2016}. In such schemes [\ref{fig:ArchitectureIII}], one detection channel corresponds to preserving path distinguishability, while another channel involves projecting system $I$ onto a superposition state that erases which-path information.

At the coarse-grained level, the detection outcome is deterministically associated with the experimental choice, and conditional statistics of $S$ exhibit interference for the erasure channel and no interference for the preservation channel, satisfying both deterministic routing and distinct conditional distributions. However, the erasure operation necessarily involves a non-unitary projection, which introduces intrinsic losses. Only a subset of the realizations of system $I$ contribute to the erasure detection channel, while the remaining events are physically filtered out by the device.

The resulting interference patterns therefore arise from conditioning on successful transmission through the erasure channel. This form of intrinsic post-selection is implemented by the measurement apparatus itself rather than by an explicit data-filtering procedure performed by the experimenters. From a probabilistic perspective, it corresponds to conditioning on a non-unit probability event, which can induce correlations between otherwise independent variables, analogous to selection effects such as Berkson’s paradox.

In typical polarization-based eraser implementations, the experimental setting is chosen with equal probability between the erasure and preservation configurations. In the preservation configuration, all events are in principle detected, whereas in the erasure configuration the projection onto a superposition polarization state necessarily filters out at least half of the incoming photons. This implies a minimal average loss rate $p_{\min} = \frac{1}{4}$. In \emph{Appendix A}, we explicitly show that this minimal loss rate already suffices to ensure compatibility between the observed conditional detection statistics and the independence assumption $X \perp C$. No additional tuning of loss mechanisms is required. The construction is further generalized to arbitrary (possibly biased) choice probabilities.

More generally, polarization-based eraser implementations rely on projecting system $I$ onto superposition states that erase which-path information within a coherent polarization subspace. The interference patterns observed in coincidence statistics therefore presuppose that coherence between the relevant polarization components is preserved prior to the erasure projection. If decoherence removes this coherence beforehand, the conditional detection asymmetries analyzed here no longer arise. The structural constraint identified in the present work thus applies specifically to implementations in which coherence is maintained up to the stage at which the erasure projection is performed.

\subsection{Violation of the independence assumption}

For completeness, we consider the remaining logical possibility allowed by the structural constraint, namely scenarios in which the independence condition $X \perp C$ is violated while the other three properties are satisfied.

Consider a pair of position-entangled photons $S$ and $I$. Photon $S$ is detected at the spatially resolved detector producing outcome $X$. Photon $I$ propagates through a Mach--Zehnder interferometer whose phase is adjusted such that the two output ports occur with equal probability. This condition ensures that the violation of the independence assumption cannot be attributed to a trivial imbalance between the detection channels, and instead reflects a genuine statistical dependence between the experimental choice and the earlier detection outcome.
The outputs are monitored by two detectors, denoted $D_1$ and $D_2$. No which-path marking is introduced, so the interferometer operates coherently and recombines the two paths unitarily. The evolution is therefore lossless, and each experimental run produces exactly one detection event at either $D_1$ or $D_2$.

Moreover, the detection outcome of $I$ uniquely determines the effective measurement context, so that the association between the detection outcome and the coarse-grained variable $D$ is deterministic. Conditional detection statistics of $S$ exhibit complementary interference patterns when conditioned on $D_1$ and $D_2$, thereby satisfying the distinct conditional distributions property. In this situation, however, the effective choice variable coincides with the detection outcome itself, namely $C \equiv D$.
Consequently, $C$ is manifestly correlated with $X$, and the independence condition $X \perp C$ is violated. 

This dependence is not surprising: the beam splitter is a passive device and does not implement an externally controlled, freely selectable experimental choice. Instead, the effective ``choice'' is generated by the intrinsic quantum randomness of the interferometer itself.

Importantly, quantum mechanics forbids the realization of an active device that would simultaneously recombine the interferometer paths coherently and allow external control over the output port. Such a device would need to implement both a unitary transformation (i.e, preserve coherent recombination) and a nontrivial projection onto selected outputs (i.e, allow external selection of output ports), which is incompatible with linear quantum dynamics. If such a device existed, it could be exploited to enable signaling between correlated systems, in contradiction with the no-signaling principle~\cite{GRW1980}.

Therefore, while scenarios violating independence can formally satisfy the other structural properties, they do not correspond to DCQE architectures, since no which-path marking is introduced on system $I$, and consequently no erasure operation is performed. The observed interference patterns simply result from coherent recombination within the interferometer.

Finally, it is worth emphasizing that arbitrary patterns can be engineered in post-selected coincidence statistics whenever detection outcomes are classically conditioned on prior measurement results. As a simple illustration, one may replace the second beam splitter by an optical switch controlled by a classical processor receiving the detection signal of system $S$. The processor routes system $I$ toward detector $D_1$ whenever $S$ is detected within a predefined region on $D_0$ forming an arbitrary figure, and toward $D_2$ otherwise. After many runs, no structure appears in the marginal distribution at $D_0$, while the coincidence distributions conditioned on $D_1$ and $D_2$ reproduce the chosen figure and its complement, respectively. This example highlights that striking conditional patterns alone do not imply any retrocausal or noncausal influence, but can arise solely from classical conditioning.

\section*{Discussion}

In this work, we have identified a simple incompatibility between four intuitive properties frequently attributed to idealized DCQE scenarios: statistical independence of the choice, absence of losses, deterministic association between the choice and detection outcome, and the existence of distinct conditional detection distributions. This incompatibility is purely probabilistic and does not rely on quantum dynamics or specific experimental implementations.

Several previous studies have emphasized that the apparent paradoxical features of DCQE experiments disappear once the physical and causal structure of the experiment is analyzed in detail \cite{Kastner2019, Fankhauser2020, Waaijer2024}. In a related perspective, it has also been argued that the apparent availability of which-path information in delayed configurations may already be constrained by the structure of the underlying entangled quantum state prior to the final detection stage \cite{Qureshi2020, Qureshi2021, Qureshi2025}.

The present analysis complements these results by identifying minimal probabilistic compatibility relations between structural properties commonly attributed to idealized DCQE scenarios, formulated directly in terms
of observable variables describing the experimental architecture. More broadly, it shows that the apparent tension between delayed choices and earlier detection outcomes arises from implicitly combining structural assumptions that cannot be simultaneously satisfied. From this perspective, the appearance of paradox in DCQE experiments can be understood as a structural illusion generated by mutually incompatible expectations rather than as evidence for any unconventional causal mechanism.

The resulting structural constraint provides a transparent explanation for why realistic DCQE architectures necessarily circumvent at least one of these intuitive assumptions. By applying this framework to representative implementations, we have shown that different experimental schemes achieve their characteristic correlations through distinct mechanisms: coarse-graining that cancels conditional interference patterns, nondeterministic routing, or intrinsic losses that induce automatic post-selection.

A central insight of the analysis is the structural role played by conditioning on selected subsets of events. In particular, polarization-based quantum eraser implementations rely on non-unitary projections that necessarily discard a fraction of the incoming particles. The resulting interference patterns emerge from coincidence statistics conditioned on successful transmission through the erasure channel. From a probabilistic perspective, this corresponds to conditioning on a non-unit-probability event, which can induce correlations between otherwise independent variables. Such selection-induced correlations are formally analogous to statistical phenomena such as Berkson’s paradox.

Importantly, the demonstration presented in \emph{Appendix~A} shows that the minimal loss rates imposed by the physical realization of erasure operations are already sufficient to ensure the existence of a joint distribution compatible with both the observed conditional statistics and the independence assumption between the delayed choice and the earlier detection outcome. No additional fine-tuning of loss mechanisms is required. This highlights that losses in DCQE experiments are not merely experimental imperfections, but constitute an essential structural ingredient enabling the observed correlations.

The framework introduced here also clarifies the status of scenarios in which the independence assumption is violated. While such setups can formally reproduce conditional interference patterns without losses or nondeterministic routing, they do not involve any controlled marking and erasure of which-path information and therefore fall outside the class of quantum eraser experiments. Their inclusion nevertheless completes the logical landscape implied by the structural constraint.

In summary, by isolating the minimal probabilistic ingredients underlying DCQE experiments, this work provides a unifying and conceptually transparent perspective that clarifies how conditional interference patterns arise in realistic implementations and removes the appearance of paradox often associated with such experiments.

\section*{Appendix A: Compatibility of losses with the independence assumption}

In this Appendix, we show that intrinsic losses present in polarization-based eraser architectures are sufficient to restore compatibility between the observed conditional detection statistics and the independence assumption $X \perp C$. We treat the general case where the experimental choice between erasure and preservation is implemented with arbitrary probability.

\subsection*{Extended detection variable}

To account for loss events, we extend the detection variable $D$ by introducing an additional outcome $L$ corresponding to cases where system $I$ is not detected. Thus, $D \in \{D_{\mathrm{erase}}, D_{\mathrm{preserve}}, L\}$. We denote by
\[
q := P(C=\mathrm{erase}), \qquad 1-q := P(C=\mathrm{preserve})
\]
the probabilities associated with the two experimental configurations. Let $p := P(D=L)$
be the total loss rate. By construction of architecture III, losses occur exclusively in the erasure configuration, while all events in the preservation configuration are detected.

\subsection*{Marginal distributions implied by independence}

Assuming statistical independence between $X$ and $C$, we have
\[
P(X,C)=P(X)P(C).
\]
Hence,
\[
P_{X,C}=
\begin{array}{c|cc}
 & \mathrm{erase} & \mathrm{preserve} \\ \hline
X & q P(X) & (1-q) P(X)
\end{array},
\]
where $P(X)$ denotes the unconditional distribution of detection positions of system $S$. The detection-choice marginal distribution follows directly from the structure of the setup:
\[
P_{C,D}=
\begin{array}{c|ccc}
 & D_{\mathrm{erase}} & D_{\mathrm{preserve}} & L \\ \hline
\mathrm{erase} & q-p & 0 & p \\
\mathrm{preserve} & 0 & 1-q & 0
\end{array}.
\]

\subsection*{Constraint from distinct conditional distributions}

Architecture III is designed such that conditional statistics of $S$ differ between the erasure and preservation channels:
\[
P(X \mid D_{\mathrm{erase}}) \neq P(X \mid D_{\mathrm{preserve}}).
\]
In the idealized limit of perfect erasure, interference is fully suppressed in coincidence with $D_{\mathrm{preserve}}$, while maximal modulation appears in coincidence with $D_{\mathrm{erase}}$. This imposes a strong asymmetry: certain subsets of detection positions $A$ satisfy
\[
P(X \in A \mid D_{\mathrm{erase}})=0,
\qquad
P(X \in A \mid D_{\mathrm{preserve}})>0.
\]
This situation maximally constrains compatibility with independence and therefore provides a worst-case scenario.

\subsection*{Compatibility with a joint distribution}

We now ask whether there exists a joint distribution $P_{X,C,D}$ compatible with the marginals $P_{X,C}$ and $P_{C,D}$ and with the conditional asymmetry above. Let us decompose probabilities according to the three outcomes of $D$.

\paragraph*{Case $D=D_{\mathrm{preserve}}$.}

Conditionally on $D=D_{\mathrm{preserve}}$, one must have $C=\mathrm{preserve}$ almost surely. The distribution of $X$ is therefore simply $P(X)$ restricted to a total probability mass $1-q$. This introduces no constraint on $p$.

\paragraph*{Case $D=D_{\mathrm{erase}}$.}

Conditionally on $D=D_{\mathrm{erase}}$, one must have $C=\mathrm{erase}$ almost surely. The total probability mass is
\[
P(D=D_{\mathrm{erase}})=q-p,
\]
which requires $p \le q$.

\paragraph*{Case $D=L$.}

Conditionally on $D=L$, one also has $C=\mathrm{erase}$ almost surely. The distribution of $X$ over loss events must absorb the remaining probability weight associated with the erasure configuration. Since the total probability mass associated with $C=\mathrm{erase}$ is $q$, and the detected erasure events carry mass $q-p$, the loss channel carries mass $p$.

To compensate for the asymmetry imposed by the conditional interference pattern in the detected erasure channel, the loss channel must contribute at least half of the probability mass associated with $C=\mathrm{erase}$. This yields the constraint
\[
p \ge \frac{q}{2}.
\]

\subsection*{Allowed range of loss rates}

Combining the constraints from the three cases, we obtain
\[
\boxed{\frac{q}{2} \le p \le q}.
\]
Therefore, a minimal loss rate $p_{\min}=\frac{q}{2}$ is sufficient to ensure compatibility between the observed conditional detection statistics and the independence assumption $X \perp C$.

\subsection*{Equiprobable choice}

When the two experimental configurations are chosen with equal probability, $q=\tfrac12$, one finds
\[
p_{\min}=\frac14.
\]
This value coincides with the unavoidable $50\%$ filtering associated with the projection required to erase which-path information in polarization-based implementations.

\subsection*{Conclusion}

This explicit construction shows that the intrinsic losses imposed by the erasure operation are not incidental imperfections, but are structurally sufficient to reconcile the observed conditional interference patterns with statistical independence of the delayed choice.

Consequently, architecture III circumvents the no-go constraint not by violating independence, deterministic routing or distinct conditional distributions, but by necessarily introducing post-selection through physically enforced losses.

\bibliographystyle{apsrev4-2}
\bibliography{main}

\end{document}